\documentclass[12pt]{article}

\usepackage[colorlinks=false]{hyperref}
\usepackage[dvips]{graphicx}
\usepackage[usenames,dvipsnames]{color}
\usepackage{epsfig}
\usepackage{rotating}
\usepackage{amssymb}
\usepackage{amsmath,amsfonts}
\usepackage[latin1]{inputenc}
\usepackage{verbatim}
\usepackage[tableposition=top,font=small,labelfont=bf,format=hang]{caption}
\usepackage{booktabs}
\usepackage{multirow}
\usepackage{geometry}
\usepackage{makecell}
\usepackage{txfonts}
\usepackage{threeparttable}
\usepackage{graphicx}
\usepackage{pgfplots}
\pgfplotsset{compat=1.18}

\newtheorem{lem}{Lemma}
\newtheorem{thm}{Theorem}
\newtheorem{prop}{Proposition}

\newtheorem{rem}{Remark}
\newtheorem{coro}{Corollary}
\newtheorem{ex}{Example}

\newcommand{\F}{\mathbb{F}}

\newenvironment{proof}{\noindent\textbf{Proof.}\quad}
{\hspace{\stretch{1}}%
\rule{1ex}{1ex}\\}

\DeclareMathOperator{\wt}{wt}

\begin{document}
\title{Double circulant LCD codes over $\mathbb{F}_2+u\mathbb{F}_2$ and their related binary LCD codes}
\author{Xingxing Xu, Bo Wu\thanks{Xingxing Xu and Bo Wu are with the Key Laboratory of Intelligent Computing Signal Processing, Ministry of Education, School of Mathematical Sciences, Anhui
		University, Hefei 230601, China. Email: xuxingxing03@163.com, wubo@ahu.edu.cn. This research (X. Xu) is supported by National Natural Science Foundation of China (12471490).}, Jon-Lark Kim\thanks{Jon-Lark Kim is with Department of Mathematics, and Institute for Mathematical and Data Sciences, Sogang University, Seoul, South Korea. Email: jlkim@sogang.ac.kr. This research (J.-L. Kim) was supported in part by the Brain Korea 21 FOUR (Fostering Outstanding Universities for Research) funded by the Ministry of
Education (MOE, Korea) and National Research Foundation of Korea (NRF)
under Grant No. 4120240415042 and by Basic Science Research Program through the National Research Foundation of Korea (NRF) funded by the Ministry of Science and ICT under Grant No. RS-2025-24534992.}
}
\date{}
\maketitle
\begin{abstract}
	Linear complementary dual (LCD) codes are linear codes that intersect
	their dual trivially. We study double circulant (Euclidean) LCD codes over
	the finite chain ring $R=\mathbb F_2+u\mathbb F_2$, where $u^2=0$.
	Using the Chinese Remainder Theorem and the factorization of $x^n-1$,
	we derive an explicit enumeration formula, with a separate treatment of
	the constituent corresponding to $x-1$, which is not
	obtained by a formal specialization of the previously known counting
	formula over $\mathbb F_q[u]/(u^k)$. Assuming Artin's primitive root
	conjecture for $2$, we show that the binary Gray images contain an
	asymptotically good family of rate $1/2$ which attains the asymptotic
	Gilbert--Varshamov bound, i.e., whose relative distance can
	approach $H_2^{-1}(1/2)$ from below. More importantly from the
	constructive viewpoint, by combining Gray images with a known extension
	construction, we obtain a binary LCD code with previously unknown parameters $[29,15,7]$
	and construct two explicit permutation-inequivalent binary LCD codes
	with parameters $[45,23,9]$ and nine binary LCD codes with the best known parameters.
\end{abstract}

\section{Introduction}

Linear complementary dual (LCD) codes introduced by Massey
\cite{Massey} are linear codes having trivial intersection with their
dual. The determination of optimal binary LCD codes, in particular the
improvement of their minimum distances for prescribed length and
dimension, remains an active problem. Bounds and constructions for
binary LCD codes have been studied in
\cite{B,GJ,KO,LM,W}, and several recent works have
obtained LCD codes from structured families such as double
circulant, quasi-cyclic, Toeplitz, and formally self-dual codes
\cite{G,H,WL,Z}. From the asymptotic point of
view, Sendrier \cite{Sendrier} showed that unrestricted binary LCD
codes meet the Gilbert--Varshamov bound. These developments motivate
the search for further algebraic constructions producing
binary LCD codes with improved parameters.

A useful approach to this problem is to investigate LCD codes over a
small finite ring and then to obtain binary LCD codes through a Gray map.
In this paper we consider
$R=\mathbb F_2+u\mathbb F_2$, $u^2=0$.
The standard Gray map
$\phi:R^n\rightarrow\mathbb F_2^{2n}$ is an
$\mathbb F_2$-linear bijection and preserves Euclidean duality in the
sense needed here \cite{WS}. Consequently, the Gray image of a
Euclidean LCD code over $R$ is again a binary LCD code. This makes
LCD codes over $R$ a natural source for explicit binary constructions.
Cyclic and self-dual codes over this ring have been
studied in \cite{BU,DGHMS,K}.

We focus on double circulant LCD codes over $R$. Their generator matrices make them suitable for explicit searches, while
their algebraic structure can be described through the Chinese Remainder Theorem (CRT). The relevant theories of quasi-cyclic and quasi-twisted codes over finite fields and
finite chain rings were developed in
\cite{J,L,LS}, and related enumeration and asymptotic
arguments for double circulant or index-two codes have been considered
over finite fields and several finite rings
\cite{AO,A,KFO,QSS,SHSS,SZQSS}.
Thus much of the general CRT framework is already available.
However, the characteristic-two case
$R=\mathbb F_2+u\mathbb F_2$ requires a separate treatment of the
constituent corresponding to $x-1$. In the counting formula of
\cite{QSS}, the displayed contribution of this constituent contains
the factor $(q-2)q^{k-1}$, which cannot simply be specialized to
$q=2$. The corresponding constituent
$\langle(1,c)\rangle$ over $R$ is Euclidean LCD precisely for
$c=0$ and $c=u$. We therefore obtain
an explicit enumeration formula for Euclidean LCD double circulant
codes over $R$. Furthermore, assuming Artin's primitive root conjecture
for $2$ which is now claimed true in~\cite{Artin_OpenAI}, we show that the binary Gray images contain an asymptotically
good family of rate $1/2$ whose relative distance can approach
$H_2^{-1}(1/2)$ from below, so that this structured family attains the
asymptotic Gilbert--Varshamov bound at rate $1/2$. The key point is a
refined counting argument which separates the codewords lying in
$uA_p^2$ from the others. Although the enumeration
argument follows the general framework developed in earlier
works, these results provide the structural basis for the binary constructions
developed
below.

The main constructive purpose of this paper is to use these double
circulant LCD codes over $R$ to obtain binary LCD codes. We take the Gray
images of double circulant codes over $R$ and apply the extension construction of
\cite{LM} to obtain further binary LCD codes. In particular, we
construct a binary LCD code with parameters $[29,15,7]$, improving the
previously reported minimum distance for binary LCD $[29,15]$ codes.
We also give an explicit binary LCD $[45,23,9]$ code obtained from this construction. 
We note that in~\cite{LM}, an optimal LCD $[45,23,d]$ code has minimum distance $8 \le d \le 10$.
In addition, an independent quasi-cyclic
search yields another binary LCD $[45,23,9]$ code. The two
$[45,23,9]$ codes have different weight distributions and hence are
permutation-inequivalent. These explicit constructions constitute
the main computational contribution of the paper.

The paper is organized as follows. Section~\ref{Prel} recalls the basic notions
on LCD codes, the Gray map, and the factorization results needed later.
Section~\ref{DC} gives the CRT decomposition and the enumeration of double
circulant LCD codes over $R$. Section~\ref{Asym} establishes the conditional
asymptotic result. Section~\ref{Gray} studies the Gray images and gives explicit
binary LCD examples. Section~\ref{New} applies a known extension construction
to obtain further binary LCD codes. Section~\ref{Conlu} concludes the paper.

\section{Preliminaries}\label{Prel}
Throughout the paper,
$$
R=\mathbb F_2+u\mathbb F_2=\{0,1,u,1+u\},\, u^2=0.
$$
The ring $R$ is a finite chain ring with maximal ideal $(u)$. For $m\ge1$, we write
$$
\mathcal R_m=\mathbb F_{2^m}+u\mathbb F_{2^m},\,
\mathcal R_m^*=\mathbb F_{2^m}^*+u\mathbb F_{2^m}.
$$
These extension rings occur as CRT constituent rings in Section~\ref{al}.

\subsection{Codes, duality, and the Gray map}
Let $\mathcal A$ be a finite commutative ring. A linear code of length $n$ over $\mathcal A$ is an $\mathcal A$-submodule of $\mathcal A^n$. For $\mathbf x,\mathbf y\in\mathcal A^n$, the \emph{Euclidean inner product} is
$$
\langle\mathbf x,\mathbf y\rangle_E=\sum_{i=1}^n x_i y_i,
$$
and the Euclidean dual is denoted by $C^\perp$. If $\mathcal A$ is equipped with a ring involution $a\mapsto\bar a$, the \emph{Hermitian inner product} is
$$
\langle\mathbf x,\mathbf y\rangle_H=\sum_{i=1}^n x_i\overline{y_i},
$$
and the corresponding dual is denoted by $C^{\perp_H}$. A code is Euclidean LCD (respectively, Hermitian LCD) when its intersection with the corresponding dual is zero. In this paper, the Hermitian duality is used only for the self-reciprocal CRT constituents $\mathcal R_{2d_i}$ and the required conjugation is specified in Section~\ref{al}.

A matrix $A$ over $R$ is \emph{circulant} if each row is obtained by a cyclic shift of the preceding row. A code with generator matrix $(I,A)$, where $I$ and $A$ are $n\times n$ matrices and $A$ is circulant, is called a \emph{double circulant} code of length $2n$ over $R$.

For a family of binary codes with parameters $[N_i,K_i,D_i]$ and $N_i\to\infty$, we use
$$
\rho=\liminf_{i\to\infty}\frac{K_i}{N_i},\,
\delta=\liminf_{i\to\infty}\frac{D_i}{N_i}.
$$
The family is asymptotically good if $\rho>0$ and $\delta>0$.

The standard Gray map over $R$ is
$$
\phi(a+ub)=(b,a+b),\, a,b\in\mathbb F_2,
$$
extended to $R^n$. It is an $\mathbb F_2$-linear bijection \cite{DGHMS}. The following known fact is the Euclidean binary specialization of \cite[Theorem~3.3]{WS} and will be used when moving from ring codes to binary codes.

\begin{prop}[\cite{WS}]\label{prop:gray}
Let $C$ be a linear code over $R$. Then
$$
\phi(C^\perp)=\phi(C)^\perp.
$$
Consequently, $C$ is Euclidean LCD if and only if $\phi(C)$ is a binary Euclidean LCD code.
\end{prop}

\subsection{Factorization and reciprocal factors}
A monic polynomial over a finite local ring is called \emph{basic irreducible} if its reduction modulo the maximal ideal is irreducible over the residue field. For a polynomial $f(x)$ with nonzero constant term, its reciprocal is defined as
$$
f^*(x)=x^{\deg f}f(x^{-1}).
$$
For odd $n$, let $C_a=\{a2^i\pmod n\mid i\ge0\}$ be the $2$-cyclotomic coset containing $a$, and let $f_{C_a}(x)$ be the corresponding irreducible factor of $x^n-1$ over $\mathbb F_2$. We collect the standard facts in the following proposition.

\begin{prop}[\cite{LN,LS,LX,SJLU}]\label{prop:factorization}
Let $n$ be an odd positive integer.
\begin{enumerate}
\item [(1)] The basic irreducible factorization of $x^n-1$ over $R$ is in one-to-one correspondence with the irreducible factorization of $x^n-1$ over $\mathbb F_2$ by reduction modulo $u$ and the corresponding factors have the same degrees.
\item [(2)] For every $a$,
$$
f_{C_a}^*(x)=f_{C_{-a}}(x).
$$
If $a\ne0$ and $d=n/\gcd(a,n)$, then $f_{C_a}$ is self-reciprocal if and only if
$$
2^j\equiv-1\pmod d
$$
for some integer $j$.
\item [(3)] All nontrivial irreducible factors of $x^n-1$ (except trivial one $x-1$) are self-reciprocal if and only if $2^j\equiv-1\pmod n$ for some $j$. If $\operatorname{ord}_n(2)$ is odd, then every nontrivial irreducible factor occurs in a reciprocal pair.
\end{enumerate}
\end{prop}

The derivative of $(x^n-1)$ is $x^{n-1}$ over $\mathbb F_2$ since $n$ is odd, then $x^n-1$ is square-free and assertion~(1) follows from Hensel lifting \cite{LS}. Assertion~(2) follows from the standard description of irreducible factors in terms of binary cyclotomic cosets \cite{LN,LX}. For assertion~(3), the first statement follows from the good-pair criterion in \cite[Lemma~3]{SJLU}. For the second statement, if $\operatorname{ord}_n(2)$ is odd, then $\operatorname{ord}_d(2)$ is odd for every divisor $d>1$ of $n$. Hence, by \cite[Lemma~3]{SJLU}, all the corresponding irreducible factors occur in reciprocal pairs.

Notice that the self-reciprocity of a cyclotomic polynomial does not force each of its irreducible factors over $\mathbb F_2$ to be self-reciprocal. For example,
$$
x^7-1=(x-1)(x^3+x+1)(x^3+x^2+1),
$$
where the two cubic factors form a reciprocal pair. This distinction is the one relevant to the CRT decomposition below.

\section{Double circulant LCD codes}\label{DC}
\subsection{Algebraic structure} \label{al}
In this subsection, we assume that $n$ is an odd integer. Every double circulant code of length $2n$ may be viewed as a code of length $2$ over the ring $A_n=R[x]/(x^n-1),$ as stated in \cite{LS}. We assume that the factorization of $x^n-1$ into basic irreducible polynomials over $R$ is of the following form: $$x^n-1= (x-1)\prod\limits_{i=1}^sg_i(x)\prod\limits_{j=1}^th_j(x)h_j^*(x),$$ where $g_i(x)$ is a monic self-reciprocal basic irreducible polynomial of degree $2d_i$ for $1\le i\le s$, and $h_j(x),h_j^*(x)$ form a pair of reciprocal basic irreducible polynomials of degree $e_j$ for $1\le j\le t.$

By using the Chinese Remainder Theorem (CRT) to break down the ring $A_n=R[x]/(x^n-1)$, we have
\begin{align*}
	A_n=\frac{R[x]}{(x^n-1)} &\simeq \frac{R[x]}{(x-1)} \oplus\Bigg(\bigoplus\limits_{i=1}^s\frac{R[x]}{(g_i(x))}\Bigg)\oplus \Bigg(\bigoplus\limits_{j=1}^t\bigg(\frac{R[x]}{(h_j(x))}\oplus \frac{R[x]}{(h_j^*(x))}\bigg)\Bigg) \\
	&\simeq R\oplus \Bigg(\bigoplus\limits_{i=1}^s\frac{\mathbb{F}_2[u,x]}{(u^2,g_i(x))}\Bigg)\oplus \Bigg(\bigoplus\limits_{j=1}^t\bigg(\frac{\mathbb{F}_2[u,x]}{(u^2,h_j(x))}\oplus \frac{\mathbb{F}_2[u,x]}{(u^2,h_j^*(x))}\bigg)\Bigg) \\
	&\simeq R\oplus \Bigg(\bigoplus\limits_{i=1}^s\mathbb{F}_{2^{2d_i}}+u\mathbb{F}_{2^{2d_i}}\Bigg)\oplus \Bigg(\bigoplus\limits_{j=1}^t\big(\mathbb{F}_{2^{e_j}}+u\mathbb{F}_{2^{e_j}}\big)\oplus\big(\mathbb{F}_{2^{e_j}}+u\mathbb{F}_{2^{e_j}}\big)\Bigg) \\
	&\simeq R\oplus \Bigg(\bigoplus\limits_{i=1}^s\mathcal R_{2d_i}\Bigg)\oplus \Bigg(\bigoplus\limits_{j=1}^t\mathcal R_{e_j}\oplus \mathcal R_{e_j} \Bigg).
\end{align*}
This decomposition naturally extends to $A_n^2$ as $$A_n^2\simeq R^2\oplus \Bigg(\bigoplus\limits_{i=1}^s\big(\mathcal R_{2d_i}\big)^2\Bigg)\oplus \Bigg(\bigoplus\limits_{j=1}^t\big(\mathcal R_{e_j}\big)^2\oplus\big(\mathcal R_{e_j}\big)^2\Bigg).$$ In particular, each linear code of length $2$ over $A_n$ can be decomposed as the "CRT sum"
\begin{align}\label{eq:1}
	C\simeq C_0\oplus \Bigg(\bigoplus\limits_{i=1}^sC_i\Bigg)\oplus\Bigg(\bigoplus\limits_{j=1}^t\big(C_j^{\prime}\oplus C_j^{\prime\prime}\big)\Bigg),
\end{align}where $C_0$ is a linear code of length 2 over $R$, $C_i$ is a linear code of length 2 over $\mathcal R_{2d_i}$ for each $1\le i\le s,$ and for each $1\le j\le t,$ $C_j^{\prime}$ and $C_j^{\prime\prime}$ are linear codes of length 2 over $\mathcal R_{e_j}$. These codes are called the constituents of $C.$

For a self-reciprocal basic irreducible factor $g_i(x)$ of degree $2d_i$, the corresponding constituent ring is identified with
$\mathcal R_{2d_i}=\mathbb F_{2^{2d_i}}+u\mathbb F_{2^{2d_i}}$.
The reciprocal action induces the conjugation
$$
\overline{a+ub}=a^{2^{d_i}}+ub^{2^{d_i}},
\,a,b\in\mathbb F_{2^{2d_i}},
$$
which defines an involution on $\mathcal R_{2d_i}$, since
$\overline{\overline{c}}=c$ for every $c\in\mathcal R_{2d_i}$.

Duality under the above CRT decomposition is standard for quasi-cyclic codes over finite chain rings \cite{LS,QSS}. In particular, the self-reciprocal factors induce Hermitian duality on the corresponding constituents, whereas the constituents associated with reciprocal pairs are interchanged under Euclidean duality. Hence, $C$ is Euclidean LCD if and only if $C_0$ is Euclidean LCD, each $C_i$ is Hermitian LCD, and
$$
C_j'\cap (C_j'')^\perp=\{0\},\,
C_j''\cap (C_j')^\perp=\{0\}.
$$
Since each constituent considered here is generated by a vector of the form $(1,c)$, these conditions can be checked directly from the inner products of their generators and reduce to unit conditions in the corresponding constituent rings. This gives the following criterion.

\begin{lem}\label{lem:1}
	Let $C=\langle(1,c)\rangle$ be a double circulant code of length $2n$ with a CRT decomposition as in (\ref{eq:1}). Let $\alpha_0=(1,c_0),$ $\alpha_i=(1,c_i),$ $\alpha_j^{\prime}=(1,c_j^{\prime})$ and $\alpha_j^{\prime\prime}=(1,c_j^{\prime\prime})$ be the generators of $C_0,C_i,C_j^{\prime}$ and $C_j^{\prime\prime},$ respectively, where $1\le i\le s$ and $1\le j\le t.$ Then
	\begin{enumerate}
		\item [(1)] $C_0$ is LCD if and only if $1+c_0^2\in R^*$;
		\item [(2)] $C_i$ is Hermitian LCD if and only if $1+c_i\bar{c_i}\in \mathcal R_{2d_i}^*$, where $\bar{c_i}={a}^{2^{d_i}}+u{b}^{2^{d_i}}$ for $c_i=a+ub\in \mathcal R_{2d_i},a,b\in\F_{2^{2d_i}}$;
		\item [(3)] $C_j^{\prime}\oplus C_j^{\prime\prime}$ is LCD if and only if $1+c_j^{\prime}c_j^{\prime\prime}\in \mathcal R_{e_j}^*$.
	\end{enumerate}
\end{lem}
\subsection{Enumeration}
Enumeration formulas for index-two quasi-twisted LCD codes over $S_{q,k}=\mathbb F_q[u]/(u^k)$ are given in \cite{QSS}. In particular, the contributions arising from the self-reciprocal constituents and reciprocal pairs specialize to
$2^{4d_i}-2^{3d_i}-2^{2d_i}$ and $2^{4e_j}-2^{3e_j}+2^{2e_j}$ respectively for $q=2$ and $k=2$.

The constituent $C_0$ corresponding to $x-1$ requires separate consideration in characteristic two. For
$$
c=c_0+c_1u+\cdots+c_{k-1}u^{k-1}\in S_{q,k},
$$
the code $\langle(1,c)\rangle$ is Euclidean LCD if and only if
$1+c_0^2\ne0$ in $\mathbb F_q$. Hence the number of admissible choices is
$(q-r_q)q^{k-1}$, where
$r_q=\#\{a\in\mathbb F_q:a^2=-1\}$.
For $q=2$ and $k=2$, this number is $2$, corresponding to $c=0$ and $c=u$.
We therefore obtain the following enumeration formula for double circulant LCD codes over
$R=\mathbb F_2+u\mathbb F_2$.

\begin{thm}\label{thm:2}
	Let $n$ be an odd integer. Assume that the factorization of $x^n-1$ into basic irreducible polynomials over $R$ is of the form $$x^n-1=(x-1)\prod\limits_{i=1}^sg_i(x)\prod\limits_{j=1}^th_j(x)h_j^*(x),$$ with $n=1+\sum\limits_{i=1}^s2d_i+2\sum\limits_{j=1}^te_j.$ Then the total number of double circulant LCD codes over $R$ is $$2\prod\limits_{i=1}^s\big(2^{4d_i}-2^{3d_i}-2^{2d_i}\big)\prod\limits_{j=1}^t\big(2^{4e_j}-2^{3e_j}+2^{2e_j}\big).$$
\end{thm}
\begin{proof}
	By Lemma~\ref{lem:1}, the LCD condition is determined constituentwise under the CRT decomposition. The constituent $C_0$ contributes $2$ choices, while the self-reciprocal constituents and reciprocal pairs contribute
	$2^{4d_i}-2^{3d_i}-2^{2d_i}$ and $2^{4e_j}-2^{3e_j}+2^{2e_j},$ respectively. Multiplying these independent contributions over all CRT constituents gives
	$$
	2\prod_{i=1}^s
	\big(2^{4d_i}-2^{3d_i}-2^{2d_i}\big)
	\prod_{j=1}^t
	\big(2^{4e_j}-2^{3e_j}+2^{2e_j}\big).
	$$
	The CRT correspondence is bijective, so this counts distinct double circulant LCD codes.
\end{proof}

By Proposition~\ref{prop:factorization}, the factorization of $x^n-1$ over $R$ falls into the following three cases.
\begin{coro} \label{coro:2}
	Let $n$ be an odd integer and $\mathrm{ord}_n(2)=m$. We can characterize the factorization of $x^n-1$ into distinct basic irreducible polynomials over $R$ as follows.
	\begin{enumerate}
		\item [(1)] If $m$ is an odd integer, then all nontrivial irreducible factors occur in reciprocal pairs. Hence $$x^n-1=(x-1)\prod\limits_{j=1}^th_j(x)h_j^*(x),$$ where $h_j(x)$ and $h_j^*(x)$ are reciprocal polynomials of degree $e_j$ and $n=1+2\sum\limits_{j=1}^te_j$. Then the total number of double circulant LCD codes over $R$ is $$2\prod\limits_{j=1}^t\big(2^{4e_j}-2^{3e_j}+2^{2e_j}\big).$$
		\item [(2)] If $m$ is an even integer and there exists an integer $i$ such that $2^i\equiv -1\pmod{n}$, all nontrivial irreducible factors are self-reciprocal. Hence $$x^n-1=(x-1)\prod\limits_{i=1}^sg_i(x),$$ where $g_i(x)=g_i^*(x)$ are self-reciprocal polynomials of degree $2d_i$ and $n=1+\sum\limits_{i=1}^s2d_i$. Then the total number of double circulant LCD codes over $R$ is $$2\prod\limits_{i=1}^s\big(2^{4d_i}-2^{3d_i}-2^{2d_i}\big).$$
		\item [(3)] In the remaining cases, we have $$x^n-1=(x-1)\prod\limits_{i=1}^sg_i(x)\prod\limits_{j=1}^th_j(x)h_j^*(x),$$ where $g_i(x)$ is a self-reciprocal polynomial with degree $2d_i$ and $h_j^*(x)$ is the reciprocal polynomial of $h_j(x)$ with degree $e_j$. Then the total number of double circulant LCD codes over $R$ is $$2\prod\limits_{i=1}^s\big(2^{4d_i}-2^{3d_i}-2^{2d_i}\big)\prod\limits_{j=1}^t\big(2^{4e_j}-2^{3e_j}+2^{2e_j}\big),$$ which has been obtained in Theorem \ref{thm:2}.
	\end{enumerate}
\end{coro}

There are several examples computed by Magma \cite{magma} about the factorization of $x^n-1$ for small $n$. These examples also illustrate the distinction between self-reciprocal irreducible factors and reciprocal pairs of irreducible factors.

\begin{ex}
	{\rm Let $n=7,23,31$. We have $\mathrm{ord}_7(2)=3,\mathrm{ord}_{23}(2)=11,\mathrm{ord}_{31}(2)=5$, then by Corollary \ref{coro:2},
		\begin{align*}
			x^7-1= &\; (x - 1)(x^3 + x + 1)(x^3 + x^2 + 1),\\
			x^{23}-1= &\; (x - 1)(x^{11} + x^9 + x^7 + x^6 + x^5 + x + 1)(x^{11} + x^{10} + x^6 + x^5 + x^4 + x^2 + 1), \\
			x^{31}-1= &\; (x - 1)(x^5 + x^2 + 1)(x^5 + x^3 + 1)(x^5 + x^3 + x^2 + x + 1)(x^5 + x^4 + x^2 + x + 1)\\&\;(x^5 + x^4 + x^3 + x + 1)(x^5 + x^4 + x^3 + x^2 + 1).
	\end{align*}}
\end{ex}

\begin{ex}
	{\rm Let $n=5,9,11$. We have $\mathrm{ord}_5(2)=4,\mathrm{ord}_{9}(2)=6,\mathrm{ord}_{11}(2)=10$ and $$2^2\equiv -1\pmod{5},\,2^3\equiv -1\pmod{9},\,2^5\equiv -1\pmod{11},$$ then by Corollary \ref{coro:2},
	\begin{align*}
		x^5-1&=(x - 1)(x^4 + x^3 + x^2 + x + 1),\\
		x^9-1&=(x-1)(x^2 + x + 1)(x^6 + x^3 + 1),\\
		x^{11}-1 &=(x - 1)(x^{10} + x^9 + x^8 + x^7 + x^6 + x^5 + x^4 + x^3 + x^2 + x + 1).
\end{align*} }
\end{ex}

\begin{ex}
{\rm Let $n=15,21,35$. By Corollary \ref{coro:2} we have
	$$x^{15}-1=(x - 1)(x^2 + x + 1)(x^4 + x + 1)(x^4 + x^3 + 1)(x^4 + x^3 + x^2 + x + 1),$$ where $x^2 + x + 1$ and $x^4 + x^3 + x^2 + x + 1$ are self-reciprocal polynomials over $R$.
\begin{align*}
	x^{21}-1= &\;(x - 1)(x^2 + x + 1)(x^3 + x + 1)(x^3 + x^2 + 1)(x^6 + x^4 + x^2 + x + 1) \\&\; (x^6 + x^5 + x^4 + x^2 + 1),
\end{align*}where $x^2 + x + 1$ is self-reciprocal polynomial over $R$.
\begin{align*}
	x^{35}-1= &\; (x - 1)(x^3 + x + 1)(x^3 + x^2 + 1)(x^4 + x^3 + x^2 + x + 1)(x^{12} + x^{10} + x^9 + \\&\; x^8 + x^7 + x^4 + x^2 + x + 1)(x^{12} + x^{11} + x^{10} + x^8 + x^5 + x^4 + x^3 + x^2 + 1),
\end{align*} where $x^4 + x^3 + x^2 + x + 1$ is self-reciprocal polynomial over $R$.}
\end{ex}
\section{Asymptotics}\label{Asym}

Asymptotic results for index-two quasi-twisted self-dual and LCD codes over finite chain rings have been obtained by random-coding arguments in \cite{QSS} and related approaches can also be found in \cite{SHSS,SZQSS}. More general factorization results over finite local rings, together with applications to self-dual and LCD codes and asymptotic coding theory, were given in \cite{KFO}. Here we restrict attention to the characteristic-two ring
$R=\mathbb F_2+u\mathbb F_2$. As seen in the preceding section, the constituent corresponding to $x-1$ requires separate treatment when $q=2$. We therefore give the corresponding asymptotic result for double circulant LCD codes over $R$.

We assume Artin's primitive root conjecture for $2$ \cite{Moree2012, Artin_OpenAI}, namely that there exist infinitely many odd primes $p$ such that $2$ is a primitive root modulo $p$. Let $p$ be such a prime. Then
$$
x^p-1=(x-1)f_p(x),
\,
f_p(x)=1+x+\cdots+x^{p-1}.
$$
Since $2$ is primitive modulo $p$, the polynomial $f_p(x)$ is irreducible over $\mathbb F_2$ and hence basic irreducible over $R$. Therefore, by the Chinese Remainder Theorem,
$$
A_p=\frac{R[x]}{(x^p-1)}
\simeq
R\oplus
\left(\mathbb F_{2^{p-1}}+u\mathbb F_{2^{p-1}}\right)
=
R\oplus\mathcal R_{p-1}.
$$
Moreover, $f_p(x)$ is self-reciprocal. Under this decomposition, an element of $A_p$ has zero
$\mathcal R_{p-1}$-component if and only if it is a multiple of
$f_p(x)$. Since $(x-1)f_p(x)=0$ in $A_p$, every such element is of
the form
$$
af_p(x)=a(1+x+\cdots+x^{p-1}),\, a\in R,
$$
and hence corresponds to a constant vector
$(a,a,\ldots,a)\in R^p$.

For $e\in A_p$, let $C_e=\langle(1,e)\rangle=\{(\mu,\mu e):\mu\in A_p\}\subseteq A_p^2$
be the corresponding double circulant code. The map
$A_p\to C_e$ given by $a\mapsto a(1,e)$ is injective, and hence
$|C_e|=|A_p|=4^p$. Since the Gray map is an $\mathbb F_2$-linear bijection, $\phi(C_e)$ is a binary code of length $4p$ and dimension $2p$, i.e., of rate $1/2$.

For $0\le x\le1/2$, let
$$
H_2(x)=-x\log_2x-(1-x)\log_2(1-x),
$$
with $H_2(0)=0$. For integers $0\le r\le N$ let $V(N,r)=\sum_{i=0}^{r}\binom{N}{i}$ be the size of a binary Hamming ball of radius $r$ in $\mathbb F_2^N$. We use the standard estimate \cite{HP}
\begin{equation}\label{eq:ball}
  V(N,\lfloor\lambda N\rfloor)\le 2^{N H_2(\lambda)}\, (0\le \lambda\le 1/2).
\end{equation}

Put $B_p=\mathbb{F}_2[x]/(x^p-1)$, so that $A_p=B_p+uB_p$.
Every $\mu\in A_p$ can be written uniquely as $\mu=\mu_0+u\mu_1$ with
$\mu_0,\mu_1\in B_p$; we identify elements of $B_p$ with vectors in
$\mathbb{F}_2^p$. Let $\mathbf{j}=1+x+\cdots+x^{p-1}=f_p(x)$, the all-one vector.
For $\mu\in A_p$ let $w_G(\mu)$ denote the Hamming weight of $\phi(\mu)\in\mathbb{F}_2^{2p}$ (the \emph{Gray weight} of $\mu$).
Since $\phi(a+ub)=(b,a+b)$ coordinatewise,
\begin{equation}\label{eq:grayweight}
  w_G(\mu_0+u\mu_1)=\mathrm{wt}(\mu_1)+\mathrm{wt}(\mu_0+\mu_1).
\end{equation}
We record three elementary facts.
\begin{itemize}
\item[(F1)] Since $f_p$ is irreducible over $\mathbb{F}_2$, the CRT gives
  $B_p\cong \mathbb{F}_2\oplus\mathbb{F}_{2^{p-1}}$,
  $b\mapsto (b(1),\, b \bmod f_p)$. The kernel of the second projection is
  $f_pB_p=\{0,\mathbf{j}\}$.
\item[(F2)] Under $A_p\simeq R\oplus\mathcal{R}_{p-1}$, since $\mathcal{R}_{p-1}$
  is local with maximal ideal $u\mathcal{R}_{p-1}$, the
  $\mathcal{R}_{p-1}$-component of $\mu=\mu_0+u\mu_1$ is a unit if and only if
  $\mu_0\bmod f_p\neq 0$, i.e., if and only if $\mu_0\notin\{0,\mathbf{j}\}$.
\item[(F3)] By \eqref{eq:grayweight}, if $\mu_0=\mathbf{j}$ then
  $w_G(\mu)=\mathrm{wt}(\mu_1)+(p-\mathrm{wt}(\mu_1))=p$, and if $\mu_0=0$ then
  $w_G(\mu)=2\,\mathrm{wt}(\mu_1)$.
\end{itemize}

Let
$\mathcal{L}_p=\{e\in A_p: C_e \text{ is Euclidean LCD}\}$. Since $f_p$ is
self-reciprocal of degree $p-1$, Corollary~\ref{coro:2}(2) gives
\begin{equation}\label{eq:Lp}
  |\mathcal{L}_p|=2\bigl(2^{2(p-1)}-2^{3(p-1)/2}-2^{p-1}\bigr)
  =2^{2p-1}\bigl(1-2^{-(p-1)/2}-2^{-(p-1)}\bigr)>2^{2p-2}\,(p\ge5).
\end{equation}

\begin{lem}\label{lem:count}
Let $D$ be an integer with $0\le D<p$, and let
\[
  \mathcal{B}(D)=\{e\in A_p:\ \phi(C_e)\ \text{contains a nonzero vector of weight at most } D\}.
\]
Then
\[
  |\mathcal{B}(D)|\le 4\,V(4p,D)+2^{p+1}\,V\!\left(2p,\lfloor D/2\rfloor\right).
\]
\end{lem}

\begin{proof}
If $e\in\mathcal{B}(D)$, there is $\mu\in A_p$, $\mu\neq0$, such that, with
$\nu=\mu e$, we have $w_G(\mu)+w_G(\nu)\le D$ (note that $\mu=0$
gives the zero codeword). Write $\mu=\mu_0+u\mu_1$ and $e=e_0+ue_1$ with
$\mu_i,e_i\in B_p$. We distinguish three cases according to $\mu_0$.

\emph{Case 1: $\mu_0=\mathbf{j}$.} By (F3), $w_G(\mu)=p>D$, a contradiction.
So this case does not occur.

\emph{Case 2: $\mu_0\notin\{0,\mathbf{j}\}$.} Let
$W_1=\{(\mu,\nu)\in A_p^2:\ w_G(\mu)+w_G(\nu)\le D\}$. Since $\phi$ is a
bijection $A_p^2\to\mathbb{F}_2^{4p}$, $|W_1|\le V(4p,D)$. Fix $(\mu,\nu)\in W_1$
with $\mu_0\notin\{0,\mathbf{j}\}$, and write $\mu=(\mu',\mu'')$,
$\nu=(\nu',\nu'')$, $e=(e',e'')$ under $A_p\simeq R\oplus\mathcal{R}_{p-1}$.
The condition $\mu e=\nu$ means $\mu'e'=\nu'$ and $\mu''e''=\nu''$. By (F2),
$\mu''$ is a unit, so $e''=\nu''(\mu'')^{-1}$ is uniquely determined, while
$e'\in R$ has at most $|R|=4$ possibilities. Hence at most $4$ elements $e$ satisfy
$\mu e=\nu$, and the elements $e$ arising in this case number at most $4\,V(4p,D)$.

\emph{Case 3: $\mu_0=0$.} Then $\mu=u\mu_1$ with $\mu_1\ne0$, and since $u^2=0$,
\[
  \nu=\mu e=u\mu_1e_0=u\nu_1,\, \nu_1:=\mu_1e_0\in B_p .
\]
By (F3), $w_G(\mu)+w_G(\nu)=2\bigl(\mathrm{wt}(\mu_1)+\mathrm{wt}(\nu_1)\bigr)\le D$,
so $(\mu_1,\nu_1)\in\mathbb{F}_2^{2p}$ has weight at most $\lfloor D/2\rfloor$. There are
at most $V(2p,\lfloor D/2\rfloor)$ such pairs. Moreover, $\mu_1\neq\mathbf{j}$,
because $\mu_1=\mathbf{j}$ would give $w_G(\mu)=2p>D$. Fix such a pair
$(\mu_1,\nu_1)$. The component $e_1\in B_p$ is unrestricted ($2^p$ choices), and
$e_0$ must satisfy $\mu_1e_0=\nu_1$ in $B_p$. Under (F1), the second component of
$\mu_1$ is nonzero since $\mu_1\notin\{0,\mathbf{j}\}$; hence the second component
of $e_0$ is uniquely determined and its first component has at most $2$
possibilities. Thus at most $2\cdot2^p=2^{p+1}$ elements $e$ correspond to each
pair, and the elements $e$ arising in this case number at most
$2^{p+1}V(2p,\lfloor D/2\rfloor)$.

Adding the bounds from Cases 2 and 3 proves the lemma.
\end{proof}

\begin{rem}
{\rm Lemma~\ref{lem:count} separates the two sources of non-uniqueness of $e$.
If the $\mathcal R_{p-1}$-component of $\mu$ is a unit (Case~2), then $e$ is almost
determined by the codeword. If it is not (Case~3), there may be as many as $2^{p+1}$
values of $e$, but then the codeword lies in $uA_p^2$, whose Gray image is the
``doubled'' code $\{(\beta,\beta)\}$; hence there are only about
$2^{2pH_2(\delta)}$, rather than $2^{4pH_2(\delta)}$, candidate codewords of
relative Gray weight at most $\delta$. Bounding the number of $e$ by $2^{p+1}$
for every codeword would only yield the relative distance $H_2^{-1}(1/4)$.}
\end{rem}

\begin{thm}\label{thm:asymp}
Assume Artin's primitive root conjecture for $2$. Let $0<\delta<H_2^{-1}(1/2)$
and put $\eta=\tfrac12-H_2(\delta)>0$. Then for every prime $p\ge5$ with
$p>2/\eta$ such that $2$ is a primitive root modulo $p$, there exists
$e\in A_p$ such that $C_e$ is a double circulant Euclidean LCD code over $R$ and
$\phi(C_e)$ is a binary LCD $[4p,2p,d]$ code with $d>4\delta p$.

Consequently, there is an infinite sequence of Gray images of double circulant
LCD codes over $R$ with rate $1/2$ and
\[
  \liminf \frac{d}{N}\ \ge\ H_2^{-1}(1/2)\approx 0.1100,
\]
that is, this family attains the asymptotic Gilbert--Varshamov bound at rate $1/2$.
\end{thm}

\begin{proof}
Since $H_2^{-1}(1/2)<1/4$, we have $\delta<1/4$; put $D=\lfloor4\delta p\rfloor$, so
$0\le D<p$. Applying Lemma~\ref{lem:count} and \eqref{eq:ball} (with
$\lambda=\delta$ for both $V(4p,\lfloor 4\delta p\rfloor)$ and
$V(2p,\lfloor D/2\rfloor)\le V(2p,\lfloor2\delta p\rfloor)$), we get
\[
  |\mathcal{B}(D)|\le 4\cdot2^{4pH_2(\delta)}+2^{p+1}\cdot2^{2pH_2(\delta)}
  =2^{2p+2-4p\eta}+2^{2p+1-2p\eta}.
\]
By \eqref{eq:Lp}, $|\mathcal{L}_p|>2^{2p-2}$ for $p\ge5$, so
\[
  \frac{|\mathcal{B}(D)|}{|\mathcal{L}_p|}<2^{4-4p\eta}+2^{3-2p\eta}.
\]
If $p>2/\eta$, then $2p\eta>4$ and $4p\eta>8$, so the right-hand side is less than
$2^{-4}+2^{-1}<1$. Hence $\mathcal{L}_p\not\subseteq\mathcal{B}(D)$, and we can choose
$e\in\mathcal{L}_p\setminus\mathcal{B}(D)$. Then $C_e$ is Euclidean LCD, so
$\phi(C_e)$ is a binary LCD code by Proposition~\ref{prop:gray}; it has length $4p$ and
dimension $2p$; and every nonzero codeword of $\phi(C_e)$ has weight at least
$D+1>4\delta p$.

For the second statement, choose $\delta_k\uparrow H_2^{-1}(1/2)$. Under Artin's
conjecture there are infinitely many admissible primes, so we can pick admissible
primes $p_1<p_2<\cdots$ with $p_k\ge5$ and $p_k>2/\eta_k$, where
$\eta_k=\tfrac12-H_2(\delta_k)$. The corresponding codes have
$d_k/N_k>\delta_k$ with $N_k=4p_k$, so $\liminf_k d_k/N_k\ge H_2^{-1}(1/2)$.
Since the rate is $1/2$ and the Gilbert--Varshamov relation is
$\rho=1-H_2(\delta)$, this is exactly the Gilbert--Varshamov bound at $\rho=1/2$.
\end{proof}

\begin{rem}
	{\rm Sendrier \cite{Sendrier} proved that unrestricted
		binary LCD codes meet the asymptotic Gilbert--Varshamov bound, whose relation between rate and relative distance is
		$\rho=1-H_2(\delta)$.
		At rate $\rho=1/2$, the corresponding relative distance is $H_2^{-1}(1/2)\approx0.1100$.
		Theorem~\ref{thm:asymp} shows that the Gray images of double circulant LCD codes over $R$
		already attain this value: for every $0<\delta<H_2^{-1}(1/2)$ they contain binary LCD codes of rate $1/2$
		and relative distance greater than $\delta$. Thus, at rate $1/2$, restricting to this structured
		family causes no loss with respect to the Gilbert--Varshamov bound.
		Figure~\ref{fig:asymptotic-comparison} illustrates the
		comparison.
		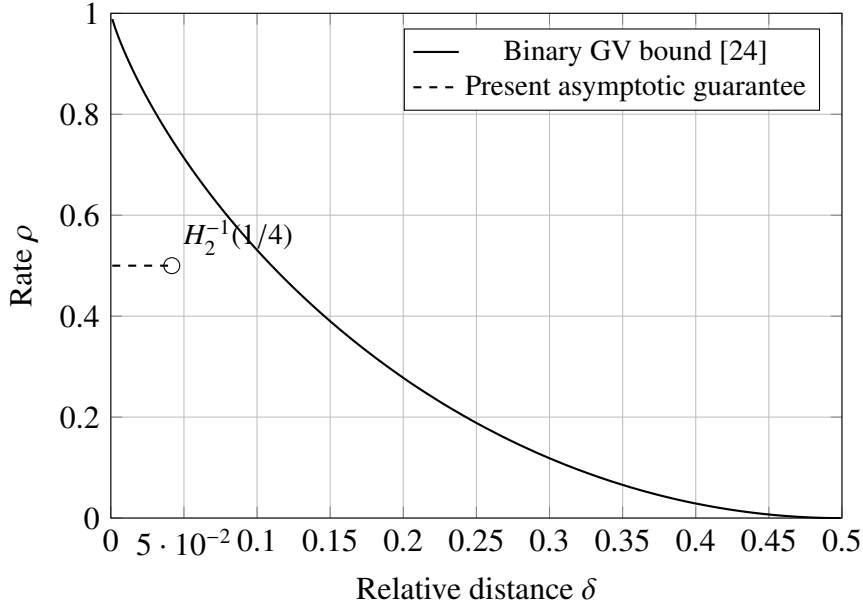
\begin{figure}[ht]
			\centering
			\begin{tikzpicture}
				\begin{axis}[
					width=0.75\textwidth,
					height=0.55\textwidth,
					xmin=0,
					xmax=0.5,
					ymin=0,
					ymax=1,
					xtick={0,0.05,0.1,0.15,0.2,0.25,0.3,0.35,0.4,0.45,0.5},
					xticklabel style={/pgf/number format/fixed,/pgf/number format/precision=2},
					xlabel={Relative distance $\delta$},
					ylabel={Rate $\rho$},
					domain=0.001:0.499,
					samples=250,
					grid=major,
					legend style={
						font=\small,
						at={(0.97,0.97)},
						anchor=north east
					}
					]

					\addplot[
					thick
					]
					{1-(-x*ln(x)/ln(2)-(1-x)*ln(1-x)/ln(2))};
					\addlegendentry{Binary GV bound \cite{Sendrier}}

					\addplot[
					thick,
					dashed
					]
					coordinates {
						(0.001,0.5)
						(0.1100279,0.5)
					};
					\addlegendentry{Present asymptotic guarantee}

					\addplot[
					only marks,
					mark=o,
					mark size=3pt,
					mark options={fill=white}
					]
					coordinates {
						(0.1100279,0.5)
					};

					\node[
					anchor=north east,
					font=\small
					] at (axis cs:0.105,0.48)
					{$H_2^{-1}(1/2)$};

				\end{axis}
			\end{tikzpicture}

			\caption{Comparison between the asymptotic relative-distance guarantee
				for the Gray images of the double circulant LCD codes considered in
				this paper and the binary Gilbert--Varshamov bound attained by
				unrestricted binary LCD codes \cite{Sendrier}. The dashed segment
				represents the range $0<\delta<H_2^{-1}(1/2)$ guaranteed by
				Theorem~\ref{thm:asymp} at rate $\rho=1/2$, and the open circle marks the
				limiting value $\delta=H_2^{-1}(1/2)\approx0.110$, which lies on the
				Gilbert--Varshamov curve.}
			\label{fig:asymptotic-comparison}
		\end{figure}
	}
\end{rem}

\section{Gray images of double circulant LCD codes over $R$}\label{Gray}
Before computing some examples, we need the following lemma.
\begin{lem}\label{lem:7}
	Let $n$ be any positive integer. If $C$ has generator matrix $G=(I,A),$ with $A=A_1+uA_2,$ where $A_1,A_2$ are binary circulant matrices of order $n$ corresponding to binary polynomials $a_1(x)$ and $a_2(x)$ respectively, and $I$ denotes the identity matrix of order $n$, then $\phi(C)$ has generator matrix $$G^{\prime}=\left( {\begin{array}{*{20}{c}}
			I&I&{{A_1}}&{{A_1}} \\
			\mathbf{0}&I&{{A_2}}&{{A_1} + {A_2}}
	\end{array}} \right).$$
\end{lem}
\begin{proof}
		Let $g_1,\ldots,g_{2n}$ be the rows of $
		G=(I,A_1+uA_2).$ Since $C$ is generated over $R$ by these rows and the Gray map is
		$\mathbb F_2$-linear, the binary code $\phi(C)$ is generated by
		$$
		\phi(g_1),\ldots,\phi(g_{2n}),
		\phi(ug_1),\ldots,\phi(ug_{2n}).
		$$
		
		Using Gray map we obtain
		$$
		\phi(G)=\big(\begin{array}{cccc}
			\mathbf{0}&I&A_2&A_1+A_2
		\end{array}\big),\,
		\phi(uG)=\big(\begin{array}{cccc}
			I&I&A_1&A_1
		\end{array}\big).
		$$
		Therefore $\phi(C)$ has generator matrix
		$$
		G'=\left( {\begin{array}{*{20}{c}}
			I&I&A_1&A_1\\
			\mathbf{0}&I&A_2&A_1+A_2
			\end{array}} \right).
		$$
		The first two block columns show that $G'$ has rank $2n$.
\end{proof}

In Table \ref{table:1}, we have collected some examples of Gray images of double circulant LCD codes obtained by random search in Magma \cite{magma}, whose parameters attain the best-known parameters. The coefficients of the polynomial
$a_1(x)$ and $a_2(x)$ in $\mathbb{F}_2[x]$, of degree $<n$, are written in decreasing powers of $x$, for example for $n=7$, the
entry $1111000$ means $x^6+x^5+x^4+x^3$. The parameters of binary LCD codes are given in the form $[4n, 2n, d]$ where $d$
is the minimum Hamming distance. The hypothesis that $n$ odd is not
needed in this subsection since we are considering LCD codes of small lengths. The entry in the rightmost column is the known minimum distance of a binary LCD $[4n,2n]$ code, obtained by looking up at the tables in \cite{B,WL,Z}.

\begin{table}[htbp]
	\centering
	\renewcommand{\arraystretch}{1.25}
	\caption{Gray images of double circulant LCD codes over $\mathbb{F}_2+u\mathbb{F}_2$, $^{\dagger}$: Best known parameters}  \label{table:1}
	\setlength{\tabcolsep}{2mm}{
	\begin{tabular}{c|c|c|c|c|c}
		\hline
		$n$ &$a_1(x)$ & $a_2(x)$ & \makecell{Parameters of \\ binary LCD codes} & \makecell{Code \\ name} &\makecell{known minimum \\ distances}\\
		\hline
		7 & $1111000$ & $1110100$ &$[28,14,7]^{\dagger}$ &$C_{28}$ & 7 \cite{B} \\
		10 & $1110010000$ & $0100100111$ & $[40,20,9]^{\dagger}$ & $C_{40}$ & 9 \cite{WL} \\
		11 & $11101000000$ & $10000110000$ & $[44,22,9]^{\dagger}$ & $C_{44}$ & 9 \cite{WL,Z} \\
		\hline
	\end{tabular}}
\end{table}

\begin{ex}{\rm
	The partial weight distribution of binary LCD $[44,22,9]$ codes generated by different matrices is listed in Table~\ref{table:5}, where $C_{44}$ appeared in Table \ref{table:1} is generated by the matrix from Lemma \ref{lem:7}, and $C$ is generated by the matrix $$G=\left( {\begin{array}{*{20}{c}}
	I_{11}&\mathbf{0}&{{A}}&{{B}} \\
	\mathbf{0}&I_{11}&{{B}}&{{A}}
\end{array}} \right)$$ in \cite[Theorem 3.8]{WL}, where $A$ and $B$ are circulant matrices generated by $ 11010011101$ and $10111110001$, respectively. From Table \ref{table:5}, since the two codes have different numbers of codewords of weight $9$, they cannot be permutation-equivalent.
}
\begin{table}[htbp]
	\centering
	\renewcommand{\arraystretch}{1.25}
	\caption{Partial weight distributions of binary LCD codes $C_{44}$ and $C$} \label{table:5}
	\setlength{\tabcolsep}{2mm}{
		\begin{tabular}{c|c|c}
			\hline
			Number $A_i$ & $C_{44}:[44,22,9]$ & $C:[44,22,9]\,\cite{WL}$ \\
			\hline
			$A_0$ & 1 & 1 \\
			$A_9$ & 198 & 132 \\
			$A_{10}$ & 616 & 693  \\
			$A_{11}$ & 1804 & 2068  \\
			$A_{12}$ & 5225 & 4763 \\
			$A_{13}$ & 12364 & 12276 \\
			$A_{14}$ & 27379 & 28457 \\
			$A_{15}$ & 55110 & 54098 \\
			$A_{16}$ & 98505 & 97449 \\
			$A_{17}$ & 162558 & 163988\\
			\hline
	\end{tabular}}
\end{table}
\end{ex}

\section{Further binary LCD constructions}\label{New}
In this section, we construct binary LCD codes from the codes in Table \ref{table:1} by using a known construction. The following part recalls a known construction of binary LCD codes by \cite[Th. 3.5]{LM}.
\begin{prop}[\cite{LM}] \label{prop:1}
	Let $C$ be an $[n,k]$ LCD code with a generator matrix $G$. Let $\mathbf{x}\in C^{\perp}$ and $C(\mathbf{x})$ be the linear code with the generator matrix $$G(\mathbf{x})=\left( {\begin{array}{*{20}{c}}
			1&\mathbf{x} \\
			\mathbf{0}&G
	\end{array}} \right).$$ If Hamming weight $\wt(\mathbf{x})$ is even, then $C(\mathbf{x})$
is a binary LCD $[n+1,k+1]$ code.
\end{prop}

Table~\ref{table:4} lists the extension vectors $\mathbf{x}$ found in our computational search. Together with Proposition~\ref{prop:1}, these vectors are used to generate further binary LCD codes. The generator matrices for $C_{28}$, $C_{40}$, and $C_{44}$ are taken from Table~\ref{table:1}.
\begin{table}[htbp]
	\centering
	\renewcommand{\arraystretch}{1.25}
	\caption{Binary LCD codes constructed by Proposition \ref{prop:1}} \label{table:4}
	\setlength{\tabcolsep}{2mm}{
		\begin{tabular}{c|c|c}
			\hline	
			\makecell{Base \\ generator matrix} & $\mathbf{x}$ & \makecell{Code \\ name}  \\
			\hline
			$G(C_{28})$ & $[0000000000000011111111111111]$ & $C_{29}$  \\
			$G(C_{29})$ & $[10000000000000000000001111111]$ & $C_{30}$ \\
			$G(C_{30})$ & $[000100001000000001000100010111]$ & $C_{31}$ \\
			$G(C_{40})$ & $[1100011000000000000000101001011111011110]$ & $C_{41}$  \\
			$G(C_{41})$ & $[01010001010100000100000100010011111101100]$ & $C_{42}$ \\
			$G(C_{42})$ & $[110011100001100000011100000111001000011011]$ & $C_{43}$ \\
			$G(C_{44})$ & $[00000000000000000000001111111111111111111111]$ & $C_{45}$ \\
			$G(C_{45})$ & $[110001100000000000000010000010100111011101011]$ & $C_{46}$ \\
			$G(C_{46})$ & $[0101010000001000000000011011011111111010001110]$ & $C_{47}$ \\
			\hline
	\end{tabular}}
\end{table}

Table~\ref{table:2} compares the parameters of these constructive binary LCD codes from Table~\ref{table:4} with the parameters or lower bounds reported in \cite{B,LM,W}. A star indicates an improvement of original results, while $^{\dagger}$ marks the parameters attaining the cited best-known distance.
\begin{table}[htbp]
	\centering
	\renewcommand{\arraystretch}{1.25}
	\caption{Constructive binary LCD codes; $^*$: improves the minimum distance, $^{\dagger}$: attains the known distance or lower bound} \label{table:2}
	\begin{threeparttable}
		\begin{tabular}{c|c||c|c||c|c}
			\hline
			Parameters & $d$ & Parameters & $d$ & Parameters & $d$ \\
			\hline
			$[29,15,7]^*$ & 6 \cite{B,LM} & $[30,16,6]^{\dagger}$ & 6 \cite{B,LM} & $[31,17,6]^{\dagger}$ & 6 \cite{B,LM} \\
			\hline
			$[41,21,8]$ &9-10 \cite{W}& $[42,22,8]^{\dagger}$ &8-10 \cite{W}& $[43,23,8]^{\dagger}$ & 8-10 \cite{W}  \\
			\hline
			$[45,23,9]^*$ &9*-10 \cite{G,W}& $[46,24,8]^{\dagger}$ & 8-10 \cite{W} & $[47,25,8]^{\dagger}$ & 8-10 \cite{W} \\
			\hline
		\end{tabular}
		\begin{tablenotes}
			\footnotesize
			\item[1] In each row, the parameters are obtained successively from the parameters in the previous column according to the construction method in Proposition~\ref{prop:1} and the vector $\mathbf{x}$ given in Table~\ref{table:4}.
		\end{tablenotes}
	\end{threeparttable}
\end{table}

From Table~\ref{table:2}, we obtain binary LCD codes with parameters
$[29,15,7]$ and $[45,23,9]$. The former exceeds the previously
reported value for binary LCD $[29,15]$ codes, while the latter
provides an explicit construction for a parameter already mentioned
in \cite{G}.
\begin{ex}\label{ex:5}{\rm
	Let $n=7, a_1(x)=x^6+x^5+x^4+x^3,a_2(x)=x^6+x^5+x^4+x^2$, we have circulant matrices
	$$A_1=\left( {\begin{array}{*{20}{c}}
			1&1&1&1&0&0&0 \\
			0&1&1&1&1&0&0 \\
			0&0&1&1&1&1&0 \\
			0&0&0&1&1&1&1 \\
			1&0&0&0&1&1&1 \\
			1&1&0&0&0&1&1 \\
			1&1&1&0&0&0&1
	\end{array}} \right), A_2=\left( {\begin{array}{*{20}{c}}
	1&1&1&0&1&0&0 \\
	0&1&1&1&0&1&0 \\
	0&0&1&1&1&0&1 \\
	1&0&0&1&1&1&0 \\
	0&1&0&0&1&1&1 \\
	1&0&1&0&0&1&1 \\
	1&1&0&1&0&0&1
\end{array}} \right),$$ $$A_1+A_2=\left( {\begin{array}{*{20}{c}}
0&0&0&1&1&0&0 \\
0&0&0&0&1&1&0 \\
0&0&0&0&0&1&1 \\
1&0&0&0&0&0&1 \\
1&1&0&0&0&0&0 \\
0&1&1&0&0&0&0 \\
0&0&1&1&0&0&0
\end{array}} \right).$$ Then the generator matrix of $\phi(C)$ is composed of block matrices, that is, $$G^{\prime}=\left( {\begin{array}{*{20}{c}}
			I_7&I_7&{{A_1}}&{{A_1}} \\
			\mathbf{0}&I_7&{{A_2}}&{{A_1} + {A_2}}
	\end{array}} \right),$$ where $\mathbf{0}$ is a $7\times 7$ zero matrix and $I_7$ is the identity matrix of order $7$. After checking by Magma, the code generated by $G^{\prime}$ is a binary LCD $[28,14,7]$ code.

Let $\mathbf{x}=(0000000000000011111111111111)\in C^{\perp}$ with $\wt(\mathbf{x})=14$, then the generator matrix of $[29,15]$ code is $$\resizebox{\textwidth}{!}{$
	G(\mathbf{x})=
	\left( {\begin{array}{*{29}{c}}
			1&0&0&0&0&0&0&0&0&0&0&0&0&0&0&1&1&1&1&1&1&1&1&1&1&1&1&1&1 \\
			0&1&0&0&0&0&0&0&1&0&0&0&0&0&0&1&1&1&1&0&0&0&1&1&1&1&0&0&0 \\
			0&0&1&0&0&0&0&0&0&1&0&0&0&0&0&0&1&1&1&1&0&0&0&1&1&1&1&0&0 \\
			0&0&0&1&0&0&0&0&0&0&1&0&0&0&0&0&0&1&1&1&1&0&0&0&1&1&1&1&0 \\
			0&0&0&0&1&0&0&0&0&0&0&1&0&0&0&0&0&0&1&1&1&1&0&0&0&1&1&1&1 \\
			0&0&0&0&0&1&0&0&0&0&0&0&1&0&0&1&0&0&0&1&1&1&1&0&0&0&1&1&1 \\
			0&0&0&0&0&0&1&0&0&0&0&0&0&1&0&1&1&0&0&0&1&1&1&1&0&0&0&1&1 \\
			0&0&0&0&0&0&0&1&0&0&0&0&0&0&1&1&1&1&0&0&0&1&1&1&1&0&0&0&1 \\
			0&0&0&0&0&0&0&0&1&0&0&0&0&0&0&1&1&1&0&1&0&0&0&0&0&1&1&0&0 \\
			0&0&0&0&0&0&0&0&0&1&0&0&0&0&0&0&1&1&1&0&1&0&0&0&0&0&1&1&0 \\
			0&0&0&0&0&0&0&0&0&0&1&0&0&0&0&0&0&1&1&1&0&1&0&0&0&0&0&1&1 \\
			0&0&0&0&0&0&0&0&0&0&0&1&0&0&0&1&0&0&1&1&1&0&1&0&0&0&0&0&1 \\
			0&0&0&0&0&0&0&0&0&0&0&0&1&0&0&0&1&0&0&1&1&1&1&1&0&0&0&0&0 \\
			0&0&0&0&0&0&0&0&0&0&0&0&0&1&0&1&0&1&0&0&1&1&0&1&1&0&0&0&0 \\
			0&0&0&0&0&0&0&0&0&0&0&0&0&0&1&1&1&0&1&0&0&1&0&0&1&1&0&0&0
	\end{array}} \right).
	$}$$ This yields a binary LCD $[29,15,7]$ code whose the partial weight distribution is listed in Table~\ref{table:3}. This parameter is
noteworthy because the previously reported value is
$6$ of binary LCD $[29,15]$ code \cite{B}, thus the present construction
shows that this value must be revised.
}
\end{ex}

\begin{ex}\label{ex:6}{\rm
	Using the same steps, we construct another binary LCD code of length $45$ with the following generator matrix:
	$$\resizebox{\textwidth}{!}{$
	G=\left( {\begin{array}{*{45}{c}}
			1&0&0&0&0&0&0&0&0&0&0&0&0&0&0&0&0&0&0&0&0&0&0&1&1&1&1&1&1&1&1&1&1&1&1&1&1&1&1&1&1&1&1&1&1 \\
			0&1&0&0&0&0&0&0&0&0&0&0&1&0&0&0&0&0&0&0&0&0&0&1&1&1&0&1&0&0&0&0&0&0&1&1&1&0&1&0&0&0&0&0&0 \\
			0&0&1&0&0&0&0&0&0&0&0&0&0&1&0&0&0&0&0&0&0&0&0&0&1&1&1&0&1&0&0&0&0&0&0&1&1&1&0&1&0&0&0&0&0 \\
			0&0&0&1&0&0&0&0&0&0&0&0&0&0&1&0&0&0&0&0&0&0&0&0&0&1&1&1&0&1&0&0&0&0&0&0&1&1&1&0&1&0&0&0&0 \\
			0&0&0&0&1&0&0&0&0&0&0&0&0&0&0&1&0&0&0&0&0&0&0&0&0&0&1&1&1&0&1&0&0&0&0&0&0&1&1&1&0&1&0&0&0 \\
			0&0&0&0&0&1&0&0&0&0&0&0&0&0&0&0&1&0&0&0&0&0&0&0&0&0&0&1&1&1&0&1&0&0&0&0&0&0&1&1&1&0&1&0&0 \\
			0&0&0&0&0&0&1&0&0&0&0&0&0&0&0&0&0&1&0&0&0&0&0&0&0&0&0&0&1&1&1&0&1&0&0&0&0&0&0&1&1&1&0&1&0 \\
			0&0&0&0&0&0&0&1&0&0&0&0&0&0&0&0&0&0&1&0&0&0&0&0&0&0&0&0&0&1&1&1&0&1&0&0&0&0&0&0&1&1&1&0&1 \\
			0&0&0&0&0&0&0&0&1&0&0&0&0&0&0&0&0&0&0&1&0&0&0&1&0&0&0&0&0&0&1&1&1&0&1&0&0&0&0&0&0&1&1&1&0 \\
			0&0&0&0&0&0&0&0&0&1&0&0&0&0&0&0&0&0&0&0&1&0&0&0&1&0&0&0&0&0&0&1&1&1&0&1&0&0&0&0&0&0&1&1&1 \\
			0&0&0&0&0&0&0&0&0&0&1&0&0&0&0&0&0&0&0&0&0&1&0&1&0&1&0&0&0&0&0&0&1&1&1&0&1&0&0&0&0&0&0&1&1 \\
			0&0&0&0&0&0&0&0&0&0&0&1&0&0&0&0&0&0&0&0&0&0&1&1&1&0&1&0&0&0&0&0&0&1&1&1&0&1&0&0&0&0&0&0&1 \\
			0&0&0&0&0&0&0&0&0&0&0&0&1&0&0&0&0&0&0&0&0&0&0&1&0&0&0&0&1&1&0&0&0&0&0&1&1&0&1&1&1&0&0&0&0 \\
			0&0&0&0&0&0&0&0&0&0&0&0&0&1&0&0&0&0&0&0&0&0&0&0&1&0&0&0&0&1&1&0&0&0&0&0&1&1&0&1&1&1&0&0&0 \\
			0&0&0&0&0&0&0&0&0&0&0&0&0&0&1&0&0&0&0&0&0&0&0&0&0&1&0&0&0&0&1&1&0&0&0&0&0&1&1&0&1&1&1&0&0 \\
			0&0&0&0&0&0&0&0&0&0&0&0&0&0&0&1&0&0&0&0&0&0&0&0&0&0&1&0&0&0&0&1&1&0&0&0&0&0&1&1&0&1&1&1&0 \\
			0&0&0&0&0&0&0&0&0&0&0&0&0&0&0&0&1&0&0&0&0&0&0&0&0&0&0&1&0&0&0&0&1&1&0&0&0&0&0&1&1&0&1&1&1 \\
			0&0&0&0&0&0&0&0&0&0&0&0&0&0&0&0&0&1&0&0&0&0&0&1&0&0&0&0&1&0&0&0&0&1&1&0&0&0&0&0&1&1&0&1&1 \\
			0&0&0&0&0&0&0&0&0&0&0&0&0&0&0&0&0&0&1&0&0&0&0&1&1&0&0&0&0&1&0&0&0&0&1&1&0&0&0&0&0&1&1&0&1 \\
			0&0&0&0&0&0&0&0&0&0&0&0&0&0&0&0&0&0&0&1&0&0&0&0&1&1&0&0&0&0&1&0&0&0&1&1&1&0&0&0&0&0&1&1&0 \\
			0&0&0&0&0&0&0&0&0&0&0&0&0&0&0&0&0&0&0&0&1&0&0&0&0&1&1&0&0&0&0&1&0&0&0&1&1&1&0&0&0&0&0&1&1 \\
			0&0&0&0&0&0&0&0&0&0&0&0&0&0&0&0&0&0&0&0&0&1&0&0&0&0&1&1&0&0&0&0&1&0&1&0&1&1&1&0&0&0&0&0&1 \\
			0&0&0&0&0&0&0&0&0&0&0&0&0&0&0&0&0&0&0&0&0&0&1&0&0&0&0&1&1&0&0&0&0&1&1&1&0&1&1&1&0&0&0&0&0
		\end{array}} \right)?$}$$ By Magma computation, the code generated by $G$ is a binary LCD $[45,23,9]$ code, whose partial weight distribution has been presented in Table \ref{table:3}.}
\end{ex}
\begin{rem}
{\rm The preprint \cite{G} mentions a binary LCD $[45,23,9]$ code, but the specific generator polynomial for that code is not listed in its Appendix B. Moreover, the generator polynomial for the $[46,23,10]$ code listed in Appendix B does not reproduce minimum distance $10$ under the polynomial-to-vector convention stated in the original text. Indeed, for
$$
f_{\rm old}(x)=x^{18}+x^{16}+x^{13}+x^{12}+x^{11}+x^{10}+x^9+x^6+x^5+x,
$$
let $a(x)=1+x^4+x^5+x^{11},$ one obtains
$$
a(x)f_{\rm old}(x)\equiv x^{17}+x^{10}+x^6+x^4\pmod{x^{23}-1}.
$$
Thus the code contains the codeword $(a(x),a(x)f_{\rm old}(x))$ of Hamming weight $4+4=8$. Consequently, puncturing one coordinate from this code cannot increase the minimum distance; hence the punctured code still has minimum distance at most $8$, and in particular cannot have minimum distance $9$.

An independent search produced a generator polynomial for the quasi-cyclic code with
$$
g(x)=1,\,f(x)=x^{16}+x^{14}+x^{13}+x^{11}+x^9+x^7+x^6+x^5+x^3+x^2+x+1.
$$
The resulting QC code is LCD with parameters $[46,23,10]$ by \cite[Theorem~4.4]{G}.

Order the coordinates as $(a_0,\ldots,a_{22}\mid b_0,\ldots,b_{22})$. Due to the cyclic symmetry of the quasi-cyclic (QC) code, puncturing any coordinate in the second block yields permutation-equivalent codes. As verified by Magma, the punctured code is a binary LCD $[45,23,9]$ code. We denote it by $C_{45}^{QC}$. The code $C_{45}$ constructed in Example \ref{ex:6} has $242$ codewords of weight $9$, whereas $C_{45}^{QC}$ has $124$. Consequently, these two binary LCD codes are not permutation-equivalent.}
\end{rem}

Based on the binary LCD codes in Table~\ref{table:1} and Table~\ref{table:2} that attain the known minimum distance but are not permutation-equivalent, Table~\ref{table:3} presents the weight distributions of these codes, where $A_i$ denotes the number of codewords of weight $i$.

\begin{table}[htbp]
	\centering
	\renewcommand{\arraystretch}{1.25}
	\caption{Partial weight distributions of the constructed binary LCD codes} \label{table:3}
	\begin{tabular}{c|c|c c c c c c c}
		\hline
		Codes &	Parameters &$A_6$ & $A_7$ & $A_8$ & $A_9$ & $A_{10}$ & $A_{11}$ & $A_{12}$\\
		\hline
		$C_{28}$ & $[28,14,7]^{\dagger}$ & & $100$ & $231$ & $378$ & $847$ & $1344$ & $1701$\\
		$C_{29}$ & $[29,15,7]^*$ & & $121$ &$329$&$609$&$1239$& $2037$ & $3003$ \\
		$C_{30}$ & $[30,16,6]^{\dagger}$ & $7$ &$163$&$428$&$827$&$1848$&$3332$ & $5096$\\
		$C_{31}$ & $[31,17,6]^{\dagger}$ & $17$ &$210$&$553$&$1156$&$2704$&$5206$ & $8476$\\
		$C_{40}$ & $[40,20,9]^{\dagger}$ &&&& $320$ & $992$ & $2180$ &$5505$\\
		$C_{42}$ & $[42,22,8]^{\dagger}$ &&& $52$ & $500$ & $1574$ & $4101$ & $10681$\\
		$C_{43}$ & $[43,23,8]^{\dagger}$ &&& $86$ & $626$ & $2005$ & $5506$ & $14661$ \\
		$C_{44}$ & $[44,22,9]^{\dagger}$ &&&& $198$ & $616$ & $1804$ & $5225$\\
		$C_{45}$ & $[45,23,9]^*$ &&&& $242$ & $770$ & $2409$ & $7251$\\
		$C_{45}^{QC}$ & $[45,23,9]^*$ &&&& $124$ & $877$ & $2754$ & $7195$\\
		$C_{46}$ & $[46,24,8]^{\dagger}$ &&& $13$ & $298$ & $978$ & $3191$ & $9644$\\
		$C_{47}$ & $[47,25,8]^{\dagger}$ &&& $27$ & $365$ & $1250$ & $4177$ & $12760$\\
		\hline
	\end{tabular}
\end{table}
\section{Conclusion}\label{Conlu}

In this paper, we studied double circulant (Euclidean) LCD codes over
$R=\mathbb F_2+u\mathbb F_2$. By means of the CRT decomposition of
$R[x]/(x^n-1)$, we obtained an explicit enumeration formula in which
the characteristic-two case is treated separately from
the previously known general counting formulas. Assuming Artin's
primitive root conjecture for $2$, we further showed that the binary
Gray images contain an asymptotically good family of rate $1/2$ whose
relative distance can approach $H_2^{-1}(1/2)$ from below; that is, this
family attains the asymptotic Gilbert--Varshamov bound at rate $1/2$.

The ring construction also leads to explicit binary LCD codes.
Combining Gray images with a known extension construction, we obtain
a binary LCD code with parameters $[29,15,7]$. We also construct an
explicit binary LCD $[45,23,9]$ code and another
$[45,23,9]$ quasi-cyclic LCD code. Their different weight distributions
show that they are not permutation-equivalent. These constructions
illustrate that double circulant LCD codes over
$\mathbb F_2+u\mathbb F_2$ provide not only a suitable family for
enumeration and asymptotic analysis, but also a useful source of
binary LCD codes with improved or refined parameters
\section*{Declarations}

\noindent\textbf{Data availability} All data produced is available from the authors upon request.  \\

\noindent\textbf{Conflict of Interest} The authors declare no conflict of interest.

\end{document}